\documentclass[11pt]{article}

\usepackage{amsmath,amsthm,verbatim,amssymb,amsfonts,amscd, graphicx}
\usepackage{graphics}
\usepackage{tikz-cd}
\usepackage{subcaption}
\theoremstyle{plain}
\newtheorem{theorem}{Theorem}

\newtheorem{proposition}{Proposition}
\newtheorem{example}{Example}

\theoremstyle{definition}

\DeclareMathOperator{\comp}{comp}
\DeclareMathOperator{\comb}{comb}
\DeclareMathOperator{\geom}{geom}
\DeclareMathOperator{\PARITY}{PARITY}
\DeclareMathOperator{\depth}{depth}

\begin{document}

\title{Comparison communication protocols}
\author{Michael R. Klug}
\maketitle

\begin{abstract}
We introduce a restriction of the classical 2-party deterministic communication protocol where Alice and Bob are restricted to using only comparison functions.  We show that the complexity of a function in the model is, up to a constant factor, determined by a complexity measure analogous to Yao's tiling number, which we call the geometric tiling number which can be computed in polynomial time.  As a warm-up, we consider an analogous restricted decision tree model and observe a 1-dimensional analog of the above results.     
\end{abstract} 


\section{Introduction} \label{sec:intro}

Communication complexity is the study of several related models of computation that measure and attempt to minimize the amount of information that needs to be communicated between different parties to compute a function.  It has been remarkably successful within computational complexity with strong lower bounds having been proven in various models of communication complexity.  Many of these results have had applications throughout theoretical computer science, including leading to lower bounds in other concrete models of computation (for example circuit depth).  In this note, we study a variation of the usual 2-party deterministic communication protocol where Alice and Bob are restricted to only using comparison functions.  As a warm-up, we study the analogous restriction of the usual deterministic decision tree model where the query node functions are restricted to being comparison functions.  This restricted decision tree model ``collapses'' in the sense that there is a simple to compute lower bound for the complexity in this model that also is an upper bound and hence determining the complexity of a function in this model is straight-forward.  Up a dimension, an analogous lower bound holds for our restricted communication complexity model and, up to a constant factor, this model also collapses in the same sense that the lower bound ends up being an upper bound.  This bound for the complexity of our model is similar to the tiling number introduced by Yao \cite{yao1979}.  We then show that this modified tiling number can be computed in polynomial time.  

We now elaborate on some of the specifics of the models under discussion although precise details will be withheld until the relevant sections.  Let $f : X \times Y \to Z$.  In the classical 2-party deterministic communication complexity model, Alice and Bob are each given half of the input, $x \in X$ and $y \in Y$, respectively, and asked to together compute $f(x,y)$ using some protocol $\mathcal{P}$.  The size of $\mathcal{P}$ is the worst-case number of bits that need to be exchanged between Alice and Bob in order to arrive at $f(x,y)$ and the complexity of $f$, denoted $C(f)$ is the minimal size of a protocol $\mathcal{P}$ that computes $f$.  For simplicity, we restrict our attention to functions of the form $f: \{0,1\}^n \times \{0,1\}^n \to \{0,1\}$ together with the lexicographical ordering on $\{0,1\}^n$.



A lower bound for $C(f)$ dating back to Yao's original article is the invariant $\chi(f)$ which is the minimum number of tilings of any $f$-monochromatic tiling of $\{0,1\}^n \times \{0,1\}^n$.  By taking any protocol $\mathcal{P}$ that computes $f$ and  labeling every element $(x,y) \in X \times Y$ by the communication history between Alice and Bob when the protocol is run on the input $(x,y)$, Yao observed that the subset of $X \times Y$ with a given communication history is not arbitrary, but rather is always combinatorial rectangles, thus leading to the bound $C(f) \geq \log \chi(f)$.  Many other lower bound techniques in this setting follow by giving lower bounds for $\chi(f)$ (for example, the size of a fooling set for $f$, the rank of $M(f)$, or the reciprocal of the discrepancy of $M(f)$ all yield lower bounds for $C(f)$ by giving lower bounds for $\chi(f)$ -- see \cite{nisan_kushilevitz}).  

We introduce here a variation of (in fact, a weakening of) the deterministic 2-party communication complexity model which we call the \emph{comparison communication complexity model}.  A comparison function is a function of the form $\theta_x : \{0,1\}^n \to \{0,1\}$ with $x \in \{0,1\}^n$ and $\theta_x(y) = 1$ if and only if $y \geq x$.  In the comparison communication complexity model, Alice and Bob's computational power is limited to the use of comparison functions (as opposed to arbitrary functions which are allowed in the classical model).  We denote the complexity of a function in this model by $C^{\comp}(f)$.  We find a variation of Yao's tiling number, using geometric rectangles in place of combinatorial rectangles (see section \ref{sec:comm}), which we denote by $\chi^{\geom}(f)$.   We prove that up to a logarithm and a constant factor $C(f)$ is equal to $\chi^{\geom}(f)$.  Motivated by this, we give a polynomial-time algorithm for computing $\chi^{\geom}(f)$ that comes from a more general consideration of finding the minimal number of rectangles needed to tile a finite planar region that can be tiled by rectangles.  

As a warm-up to this, we study a variation on the decision tree model, where instead of allowing arbitrary coordinate query functions, the query functions are comparison functions.  We denote the complexity of a function with respect to this model by $D^{\comp}(f)$.  We find an explicit formula for $D^{\comp}(f)$ in terms of a simple to compute complexity measure $\mu(f)$, which is essentially a one-dimensional version of $\chi^{\geom}$.

In section \ref{sec:tree}, we discuss the comparison decision tree model, give a formula for computing the complexity of a function in this model which can be computed in polynomial time, compare the complexity of a few functions in this model with the complexity of those functions in the classical decision tree model, and show that many functions have large complexity for this model.  In section \ref{sec:comm}, we introduce and prove that analogous results for the comparison communication complexity model where we show that the complexity of a function in this model is $\Theta(\chi^{\geom}(f))$ where $\chi^{\geom}$ is a complexity measure that is analogous to Yao's tiling number $\chi$ but with geometric rectangles used in place of combinatorial rectangles. 


\section{Comparison decision trees}\label{sec:tree}

In the classical decision tree model of computation, one considers decision trees $\mathcal{T}$ that compute a given boolean function $f: \{0,1\}^n \to \{0,1\}$ where each non-leaf vertex $v$ in $\mathcal{T}$ is labeled by a query function of the form $\pi_i : \{0,1\}^n \to \{0,1\}$ for some $1 \leq i \leq n$, where $\pi_i$ is projection onto the $i$-th coordinate.  The decision tree complexity of $f$, denoted by $D(f)$ is the worst-case number of queries needed for a decision tree $\mathcal{T}$ to compute $f$, minimized over all such decision trees $\mathcal{T}$ that compute $f$ (see \cite{buhrman_wolf}). 

We consider a variation of this model which we call the comparison decision tree model where instead of having the vertices of a decision tree $\mathcal{T}$ labeled by projection functions, each such $v$ is labeled by a comparison function $\theta_x : \{0,1\}^n \to \{0,1\}$ where $x \in \{0,1\}^n$ and $\theta_x(y) = 1$ if and only if $y \geq x$ where here we are using the lexicographical ordering on $\{0,1\}^n$.  We additionally we allow vertices to be labeled by the constant 0 function (which can be thought of as asking if the input is less than $0^n$) -- note that the constant 1 function is given by $\theta_{1^n}$.  We call such trees comparison decision trees and define $D^{\comp}(f)$ analogously to $D(f)$ but with comparison decision trees in place of decision trees.  

We have the following immediate upper bound, identical to the trivial upper bound for the classical decision tree complexity:

\begin{proposition} \label{prop:trivial_bound_decision}
For all $f: \{0,1\}^n \to \{0,1\}$, 
$$
D^{\comb}(f) \leq n
$$
\end{proposition}

\begin{proof}
We construct a comparison decision tree $\mathcal{T}$ that computes $f$ in $n$ queries on all inputs.  The first query determines if the input, considered as the binary representation of a number, is in $[0, 2^{n-1}]$ or $(2^{n-1}, 2^n]$.  The next round of queries further determines which half of each of these intervals the input is in.  Continuing in this fashion, after $n$ queries, we know exactly what the input is and we label those leaves of $\mathcal{T}$ appropriately according to $f$.  
\end{proof}

We now introduce a complexity measure of boolean functions that, for comparative decision trees, completely captures the complexity.  For $f: \{0,1\}^n \to \{0,1\}$, let $\mu(f)$ denote the number of maximal connected blocks of zeroes and  of ones that occur in the length $2^n$-list of the values of $f$ on the $2^n$ different inputs, ordered with the lexicographical ordering.     

The following result gives a straightforward formula for the complexity of the comparison decision tree model and shows that $D^{\comp}(f)$ is computable in polynomial time given $f$.  Thus in some sense, the comparative decision tree model collapses to something not so interesting.  In contrast, computing $D(f)$ is known to be NP-complete \cite{hyafil_rivest}.  

\begin{theorem} \label{thm:char_dec_tree}
For all $f: \{0,1\}^n \to \{0,1\}$, 
$$
\lceil \log \mu(f) \rceil = D^{\comp}(f)
$$
\end{theorem}

\begin{proof}
We begin by proving that $\lceil \log \mu(f) \rceil \leq D^{\comp}(f)$
The key observation is that whenever $f(x) \neq f(x+1)$, then, for any comparison decision tree $\mathcal{T}$ that computes $f$, there must be some vertex in $\mathcal{T}$ with the query function $\theta_{x}$.  Therefore, any such $\mathcal{T}$ computing $f$ must have at least $\mu(f) - 1$ query vertices.  

Since the underlying tree of $\mathcal{T}$ is a full binary tree (i.e., each non-leaf vertex has exactly two children), $\mu(f)$ is less than or equal to the number of leaves of $\mathcal{T}$.  Now the fact that $\lceil \log \mu(f) \rceil \leq D^{\comp}(f)$ follows since the number of leaves of $\mathcal{T}$ is less than or equal to $2^{\depth(\mathcal{T})}$ and $D^{\comp}(f)$ is by definition the minimum $\depth(\mathcal{T})$ for all $\mathcal{T}$ that compute $f$.  

For the other direction, namely that $\lceil \log \mu(f) \rceil \geq D^{\comp}(f)$, note that we can construct a comparative decision tree by a binary search as in the proof of Proposition \ref{prop:trivial_bound_decision} to determine which of the $\mu(f)$ blocks the input is in.  This requires at worst $\lceil \log \mu(f) \rceil$ queries.  
\end{proof}


The following two examples show that $D(f)$ and $D^{\comp}(f)$ can differ arbitrarily.  

\begin{example}
For $1^n \in \{0,1\}^n$, consider the comparison function $\theta_{1^n}$.  Then $D^{\comp}(\theta_x) = 1$, whereas $s_{1^n}(\theta_{1^n})$ and therefore $s(\theta_{1^n}) = n$ (where here $s$ denotes the sensitivity), and therefore $D(\theta_{1^n})=n$.  
\end{example}

\begin{example}
Consider the projection function $\pi_n : \{0,1\}^n \to \{0,1\}$ which gives the last coordinate.  Then $D(\pi_n) = 1$, however $\mu(\pi_n) = 2^n$ and therefore, $D^{\comp}(f) = n$.  
\end{example}

In the classical decision tree model, almost all boolean functions have maximal decision tree complexity (i.e., the limit as $n$ goes to infinity of the probability that a random boolean function $f: \{0,1\} \to \{0,1\}$ has $D(f) = n$ goes to 1 -- see, for example \cite{buhrman_wolf}).  We now show that this is not the case in the comparison decision tree model and in fact the analogous limit is equal to $1/2$:

\begin{proposition}
For all natural numbers $n$, we have
$$
\# \{f: \{0,1\} \to \{0,1\} : D^{\comp}(f) = n \} = 2^{2^n -1}
$$
\end{proposition}

\begin{proof}
By Theorem \ref{thm:char_dec_tree}, we have
$$
\# \{f: \{0,1\} \to \{0,1\} : D^{\comp}(f) = n \} = \# \{f: \{0,1\} \to \{0,1\} : \mu(f) > 2^{n-1} \}
$$
Noting that 
$$
\# \{ f: \{0,1\}^n \to \{0,1\} : \mu(f) = k \} = 2 \binom{2^n - 1}{k - 1}
$$
we have 
\begin{align*}
\# \{ f: \{0,1\}^n \to \{0,1\} : \mu(f) > 2^{n-1} \} &=  \sum_{k = 2^{n-1}}^{2^n - 1}{2 \binom{2^n - 1}{k - 1}} \\ &= 2^{2^n - 1}
\end{align*}
\end{proof}

\section{Comparison 2-party deterministic communication complexity}\label{sec:comm}

We introduce here a variation of the classical 2-party deterministic communication complexity model which we will call the comparative (2-party deterministic communication complexity) model.  We will introduce a quantity analogous to $\chi$ which we will denote by $\chi_\text{geo}$ (short for ``geometric'') and we will prove the analog of the result of Yao that $D(f) \log \chi(f)$ in our setting.  

We now define the 2-party comparison communication complexity model which is a restriction of the classical 2-party communication complexity model.  In the classical model, we are given a function $f : \{0,1\}^n \times \{0,1\}^n \to \{0,1\}$ and we consider communication complexity protocols $\mathcal{P}$ that compute $f$, where each non-leaf vertex $v$ of $\mathcal{P}$ is either an arbitrary function $a_v : \{0,1,\}^n \to \{0,1\}$ or $b_v : \{0,1,\}^n \to \{0,1\}$ that operates on Alice's or Bob's inputs respectively.  The cost of such a protocol $\mathcal{P}$ is the worst case number of edges that must be traversed over all inputs, and the communication complexity of $f$, denoted $C(f)$, is the minimization over all protocols $\mathcal{P}$ that compute $f$ of the cost of $\mathcal{P}$.   We modify the allowed protocols to define the comparison model by only allowing protocols $\mathcal{P}$ where all of the vertex functions are of the form $a_v= \theta_{x_v}$ or $b_v = \theta_{x_v}$ (operating on Alice's or Bob's inputs respectively) for some $x_v \in \{0,1\}^n$ (where, as in the previous section, $\theta_{x_v}(y) = 1$ if and only if $y \geq x$). Additionally, as in the previous comparison decision tree model, we allow also for the vertices to be labeled by the constant 0 function (note the constant one function is $\theta_{1^n}$.  We define the cost of a comparative communication protocol and the comparative communication complexity of $f$, denoted $C^{\comp}(f)$ in direct analogy with the classical case -- namely the minimum over all comparison protocols $\mathcal{P}$ that compute $f$ of $\depth(\mathcal{P})$.   

Since any a comparative communication protocol is also classical communication protocol, we always have 
$$
C(f) \leq C^{\comp}(f)
$$

We have the following upper bound:

\begin{theorem}
For all $f : \{0,1\}^n \times \{0,1\}^n \to \{0,1\}$, we have 
$$
C^{\comp}(f) \leq 2n + 1
$$
\end{theorem}

\begin{proof}
We will reason as usual with Alice and Bob, however, it is important to bear in mind that they are no longer all powerful (as in the the classical model).  In particular, even once Bob knows Alice's input, he might still need some time using comparison functions to determine the output.

Alice begins by telling Bob her input by first telling Bob if it represents a number in the interval $[0,2^{n-1}]$ or not.  If yes, she then says whether her input represents a number in the interval $[0, 2^{n-2}]$ or not.  If no, she says whether her input is in the interval $(2^{n-1}, 2^{n-1} + 2^{n-2}]$ or not.  All this is accomplished using appropriate comparison functions and after the application of at most $n$ comparison functions, Alice will have told Bob her input.  

Now through an analogous process, Bob can determine his input by a similar process and then the leaf of the resulting tree can be labeled by the appropriate constant.  
\end{proof}

A subset $R \subset X \times Y$ is called a combinatorial rectangle if $R = A \times B$ for some sets $A \subset X$ and $B \subset Y$.  A tiling of $X \times Y$ is a partition of $X \times Y$ into a disjoint collection of combinatorial rectangles.  Given a function $f : X \times Y \to Z$, a combinatorial rectangle $S \subset X \times Y$ is $f$-monochromatic if $f$ is constant on $S$.  A tiling is $f$-monochromatic if all of the combinatorial rectangles in the tiling are $f$-monochromatic. 

In contrast to combinatorial rectangles, if $X$ and $Y$ are both given fixed total orderings $<_X$ and $<_Y$ respectively, then a geometric rectangle is subset $R = A \times B \subset X \times Y$ with the property that $A = \{ a \in A : a_\text{mi} \leq a \leq a_\text{max} \}$ and such that $B = \{ b \in B : b_\text{mi} \leq b \leq b_\text{max} \}$ for some $a_\text{min},a_\text{max} \in A$ and $b_\text{min}, b_\text{max} \in B$.  This terminology of geometric rectangles comes from the fact that such sets are honest connected rectangles (in the everyday usage sense) when considered as subsets of the matrix $M_{<_X, <_Y}(f)$ whose $(i,j)$-entry is the value $f(a_i, a_j)$ (here assuming $X$ and $Y$ are finite).  

Specializing to our familiar setting with $f: \{0,1\}^n \times \{0,1\}^n \to \{0,1\}$, a lower bound for the classical 2-party communication complexity of $f$, $C(f)$, dating back to Yao's original article introducing communication complexity is the invariant $\chi(f)$ which is defined as the minimum number of tilings of any $f$-monochromatic tiling of $\{0,1\}^n \times \{0,1\}^n$.  By taking any protocol $\mathcal{P}$ that computes $f$ and  labeling every entry of element $(x,y) \in \{0,1\}^n \times \{0,1\}^n$ by the communication history between Alice and Bob when the protocol is run on the input $(x,y)$, Yao observed that the subset of $X \times Y$ with a given communication history is not arbitrary, but rather is always an $f$-monochromatic combinatorial rectangle, thus leading to the bound $D(f) \geq \log \chi(f)$.  

We now prove an analogous result in the setting of comparison communication complexity with geometric rectangles in place of combinatorial rectangles. Let $\chi^{\geom}(f)$ denote the minimal number of $f$-monochromatic geometric rectangles needed to tile $\{0,1\}^n \times \{0,1\}^n$, where here and from now on we are using the lexicographical ordering on $\{0,1\}^n$.  The proof is completely analogous to the proof of Yao's result in \cite{yao1979} -- see also Exercise 1.8 of \cite{nisan_kushilevitz}.  

\begin{theorem} \label{thm:cc_lb}
For all $f : \{0,1\}^n \times \{0,1\}^n \to \{0,1\}$, we have 
$$
\log \chi^{\geom}(f) \leq C^{\comp}(f)
$$
\end{theorem}

\begin{proof}
Let $\mathcal{P}$ be a communication protocol that computes $f$.  We prove that for each vertex $v$ in $\mathcal{P}$, the set of all inputs that go through $v$ when given to $\mathcal{P}$ forms a geometric rectangle.  For each leaf vertex, these rectangles are then $f$-monochromatic.  

We proceed by induction on the depth of the vertex $v$.  For the root vertex, the corresponding set is all of $\{0,1\}^n \times \{0,1\}^n$ and hence a geometric rectangle.  Given a vertex $v$ with depth greater than 0, we let $v'$ denote the parent vertex of $v$.  Let $R_v$ and $R_{v'}$ denote the sets of all inputs that pass through $v$ and $v'$.  By the inductive hypothesis, we have $R_v = [x_{\min}, x_{\max}] \times [y_{\min}, y_{\max}]$ for some $x_{\min}, x_{\max}, y_{\min}, y_{\max} \in \{0,1\}^n$.  Without loss of generality, we assume that the vertex $v'$ is one of Alice's vertices and we let $\theta_{x'}$ denote the corresponding query function.  There are three possibilities: (1) $x' < x_{\min}$, (2) $x' \geq x_{\max}$, or (3) $x_{\min} \leq x_{v'} \leq x_{\max}$.  In case (1), $R_v$ is empty.  In case (2), $R_v = R_{v'}$.  In case (3), $R_v = [x_{\min}, x'] \times [y_{\min}, y_{\max}]$.  In any case, $R_{v'}$ is again a geometric rectangle.  

Therefore, we obtain from $\mathcal{P}$ a geometric tiling of $\{0,1\}^n \times \{0,1\}^n$ with the number of tiles being equal to the number of leaves of $\mathcal{P}$.  Since the number of leaves of $\mathcal{P}$ is at most $2^{\depth(\mathcal{P})}$, then $2^{\depth(\mathcal{P})} \geq \chi^{\geom}(f)$.  Since this holds for all protocols $\mathcal{P}$ computing $f$, the result follows.  
\end{proof}

Just as Yao's tiling number lower bound extends trivially to communication protocols for computing relations (as used for example in proving circuit depth lower bounds), the above result extends as well to comparison communication protocols for relations.

We now give an example illustrating that $C(f)$ and $C^{\comb}(f)$ can differ arbitrarily.  

\begin{example}
Let $\PARITY_{2n} : \{0,1\}^n \times \{0,1\}^n \to \{0,1\}$ with $\PARITY_{2n}(x,y)$ equal to the party of the number of 1's in $(x,y)$.  Then the protocol where Alice sends Bob the parity of her input and Bob then adds that to the parity of his input shows that $C(\PARITY_{2n}) = 2$.  For contrast, we have $\chi^{\geom}(\PARITY_{2n}) = 2^{2n}$ and therefore $C^{\comp}(\PARITY_{2n}) \geq 2n$.  

Note that this example shows that the analog of the log-rank conjecture fails for the comparative communication complexity model since the rank of $\PARITY_{2n}$ is 2 for $n \geq 2$.  
\end{example}

\begin{theorem} \label{thm:cc_main}
For all $f : \{0,1\}^n \times \{0,1\}^n \to \{0,1\}$, we have $C^{\comp}(f) = \Theta( \log \chi^{\geom}(f))$.
\end{theorem}

\begin{proof}
From Theorem \ref{thm:cc_lb}, it suffices to show that $C^{\comp}(f) = \Omega(\log \chi^{\geom}(f))$.  Suppose we are given a geometric $f$-monochromatic tiling of $\{0,1\}^n \times \{0,1\}^n$, $T$ where $T$ has $\chi^{\geom}(f)$ tiles.  We now construct a comparison protocol $\mathcal{P}$ for $f$ using $T$.  

Visualizing $T$ as a tiling of the matrix $M(f)$, we see that the tops and bottoms of the rectangles extend to lines, of which there are at most $\chi^{\geom} + 1$.  The comparison protocol begins with Alice using comparisons in a binary search fashion to tell Bob between which two horizontal lines constructed from $T$ her input is between.  This requires $\lceil \log \chi^{\geom}(f) \rceil$ comparison queries.  Similarly, Bob can determine the tile in $T$ that the input is in by using $\lceil \log \chi^{\geom}(f) \rceil$ comparison queries corresponding to the vertical lines coming from extending the tiles of $T$.  After this, the leaf is labeled with the appropriate output corresponding to if the corresponding tile $t \in T$ is a 0-rectangle or a 1-rectangle.  Thus in total, the cost of this protocol $\mathcal{P}$ is $2 \lceil \log \chi^{\geom}(f) \rceil$.
\end{proof}

In \cite{AhoUllmanYannakis}, Aho, Ullman, and Yannnkakis proved that $D(f) = O(\log^2 \chi(f))$.  In \cite{nisan_kushilevitz}, the question is raised if this square factor can be removed -- namely, if $D(f) = O(\log \chi(f))$ (which would imply $D(f) = \Theta( \log \chi(f) )$ are in fact equal), however this was disproved in \cite{goos2018deterministic}.  The question of exactly how small the exponent can be made is still open \cite{rao_yehudayoff}.  In light of this, theorem \ref{thm:cc_main} highlights the distinction between $\chi$ and $\chi^{\geom}$ by showing that in the context of comparison communication complexity and $\chi^{\geom}$, the exponent can be taken to be 1.  

Komlós proved that a random (0,1)-matrix of size $m \times m$ has rank $m$ (over $\mathbb{Q}$) with probability tending to $1$ as $m$ goes to infinity \cite{komlos}.  Therefore, using the rank lower bound for classical communication complexity \cite{mehlhorn}, almost all functions are hard from the perspective of classical communication complexity (i.e., if $f: \{0,1\}^n \times \{0,1\}^n \to \{0,1\}$ then $C(f) = n+1$).  Since $C(f) \leq C^{\geom}(f)$, this implies that almost all functions $f : \{0,1\}^n \times \{0,1\}^n \to \{0,1\}$ have $C^{\comp}(f) \geq n$.  We are not sure to what extent this can be improved.  

The quantity $\chi^{\geom}(f)$ can be computed in polynomial-time given $f$ by finding the minimal number of rectangles needed to tile the planar regions given by the region of the matrix $M(f)$ consisting of zeroes and the the region of the matrix $M(f)$ consisting of ones.  These quantities can be found in polynomial time using for example the algorithms for finding minimal rectangulations of planar regions that can be built from rectangles as in section 3 of \cite{eppstein_review} and the references therein.

\bibliographystyle{plain} 
\bibliography{main}

\end{document}